\documentclass[interact]{article}
\usepackage{graphicx}
\usepackage{epsfig}
\usepackage{amsmath}
\usepackage{amsthm}
\usepackage{amssymb}

\usepackage{adjustbox}
\usepackage{color}
\usepackage{natbib}
\usepackage{pifont}

\UseRawInputEncoding

\newcommand{\bm}[1]{\mbox{\boldmath $#1$}}

\newtheorem{lemma}{Lemma}

\usepackage{bm}
\newcommand{\n}[1] {\mbox{\boldmath{$#1$}}}

\title{On Model Prior Probabilities in Variable Selection}

\author{ James Berger, Duke University and Texas A\&M University \\
Gonzalo Garc\'ia-Donato, University of Castilla-La Mancha \\
Luis Pericchi, University of Puerto Rico, Rio Piedras}

\date{\today}
\begin{document}
\maketitle

\begin{abstract}
	For many years it was routine to use equal model prior probabilities in Bayesian
	model uncertainty analysis. At least twenty years ago it became clear that this was problematic, as it tended to favor models that were far too large within the increasingly vast model spaces being explored in genomics and other fields. A popular replacement was to adopt
	a suggestion of Harold Jeffreys for the variable selection problem in which a total
	of $k$ possible variables are being considered for inclusion in the model: give the collection of all models containing $d$ variables
	($d=0,\ldots,k$) prior probability $1/(k+1)$ and then divide this prior probability
	equally among the models in the collection. Many other choices of model
	prior probabilities that impose severe parsimony have also been introduced. We begin
	by reviewing the problems with using equal model prior probabilities and then discuss some serious problems with the Jeffreys choice. Finally, we introduce
	and study a number of objective alternative choices of model prior probabilities, from both numerical and theoretical
	perspectives. The final choice of model prior probabilities is found to depend on the degree of sparsity that one requires.
\end{abstract}

\section{Introduction and motivation}
Consider the variable selection problem, in which variables $(x_1, x_2, \ldots, x_k)$ can or cannot be included in a model. Let $\gamma_i = 1$ if variable $x_i$ is
to be included in the model and $\gamma_i = 0$ otherwise. Then $\bm\gamma =(\gamma_1,
\ldots \gamma_k)$ is the vector of zeroes and ones that defines which variables are
in the model $M_{\bm  \gamma}$, with $M_{ \bm 0}$ being the null model and $M_{\bm 1}$
being the full model. Note that there are $2^k$ possible models and that the dimension
of $M_{\bm \gamma}$ is $|\bm\gamma|= \sum \gamma_i$. We only consider prior specifications that depend on the dimensionality of $\bm\gamma$, that is, $Pr(M_{\bm\gamma})$ is a function of $|\bm\gamma|$, or equivalently, \begin{equation}\label{eq.Mg.dim}
Pr(M_{\bm\gamma})=\frac{1}{\binom{k}{|{\bm\gamma}|}} {\mathfrak M}_k(|{\bm\gamma}|),
\end{equation}
where ${\mathfrak M}_k(d)$ is the probability assigned to the set of models such that $|{\bm\gamma}|=d$. Clearly ${\mathfrak M}_k(d)$  defines a probability mass function over $d\in\{0,\ldots,k\}$.

Historically, the first two considered objective choices of model prior probabilities $
Pr(M_{\bm \gamma}) $ were {\em Uniform} probabilities
\begin{equation}\label{Uniform}
Pr(M_{\bm \gamma})  = 2^{-k}
\end{equation}
and {\em Jeffreys} probabilities:
\begin{equation} \label{General} Pr(M_{\bm \gamma}) 
= \int_0^1 p^{|\bm\gamma|}
(1-p)^{k-|\bm\gamma|}\, \pi(p)dp,
\end{equation}
with $ \pi(p)=\mbox{Un}(p\mid 0,1)$, 
leading to
\begin{equation} \label{Jeff} Pr(M_{\bm \gamma}) 
=
Beta(1+|\bm\gamma|,1+k-|\bm\gamma|)= \frac{|\bm\gamma|! \ (k-|\bm\gamma|)!}{(k+1)!} \,. 
\end{equation}
One way in which the Jeffreys probabilities can be explained is to think of $p$ as the prior inclusion probability of a variable (assumed to be common across variables); that the inclusions of different variables are assumed to be independent; and that $p$ is
assigned the natural objective uniform prior density. These probabilities  were recommended by \cite{Jeff:61}, although
the argument therein was to give each potential model size an equal probability of $1/(k+1)$, and then divide this probability up equally among all models of a given size;
it is easy to show that this gives the same answer as (\ref{Jeff}).

The uniform probabilities also arise from saying that the prior inclusion
probability for each variable is $p=0.5$ and that they are independent. Thus
there is no penalty for including variables, which can be problematical, particularly if $k$ is large. For instance, suppose one is testing 1,000,000 genes to detect for gene expression related to some disease. The uniform probability assignment implies that there is an extremely large prior probability that the number of influential genes is around 500,000 (e.g., the models, for which the number of positive genes is in the interval [499,000; 501,000], have total probability of about 0.95); this is typically very
far from actual prior opinions.  This issue with
uniform probabilities was also described as a lack of {\em multiplicity control}
in \cite{Scott:Berg:05,Scott:Berg:10}. It is also shown therein that the Jeffreys
probabilities do induce multiplicity control.

A study of gene association involving $k=4088$ genes \citep{GarDonSteel21} illustrated the dramatic difference in the use of the two priors. With the Jeffreys choice, only one gene is clearly positive (posterior inclusion probability of 0.97), while the great majority are clearly ruled out (97\% have inclusion probability below 0.005). With the Uniform prior, the posterior distribution tends to highly mimic the prior, suggesting that about half of the genes may be important (with most of the posterior inclusion probabilities
hovering around 0.5).
For other discussions of the choice of prior probabilities of models, see
\cite{Geor:Mccu:93,Geor:Mccu:97}, 
\cite{Clyd:Desi:Parmi:96}, 
\cite{Smit:Kohn:96}, 
\cite{Raft:Madi:Hoet:97}, 
\cite{Scott:Berg:05,Scott:Berg:10}, and 
\cite{Ley:Steel:09}. 

One concern with the Jeffreys probabilities is that they give large models as much weight
as small models. For instance, the null model and the full model both have equal prior
probability of $1/(k+1)$. Typically, some type of parsimony is desired in which smaller models are given more weight than large models. Indeed, use of Jeffreys probabilities has been found to produce disturbing results in application.
For instance, in the infant obesity study given in  \cite{Zuetal11} that involved $k=47$ possible variables, use of the Jeffreys  probabilities resulted in the very surprising result that the highest posterior
probability model was the full model. This is certainly not compatible with 
prior scientific beliefs in the obesity study and presumably arose because
the full model had the largest prior probability of 1/48. (The details of this application are given in the Appendix, so as
not to disrupt the flow of the paper.)

Thus we are faced with problems with both the uniform probabilities and the
Jeffreys probabilities.
In Section \ref{sec.pars}, we present some objective choices of model prior probabilities that
seek to overcome these problems. The suggested possibilities will be evaluated in the remainder of the paper, first numerically and then theoretically. 

The use of the word `objective' in this introduction deserves a comment. The Uniform and Jeffreys model prior probabilities can certainly be called objective, but we have seen that they have serious
problems, centered around a lack of parsimony. At the other extreme, the literature is replete with efforts to induce
severe parsimony, but we would not label those efforts as being objective. (Some of the efforts most closely related to the proposals studied in this paper are discussed in
Section \ref{sec.toopars}.)
What we are
seeking is choices of model prior probabilities that incorporate necessary parsimony
but only minimally so, in order to reasonably view them as being objective.


\section{More parsimonious prior probabilities}
\label{sec.pars}
Apart from Uniform and Jeffreys' we consider six additional choices of model prior probabilities, the first three
from the literature.
\subsection{Harmonic prior probabilities}
\cite{Berg:Peri:96a} mentioned the possibility of using the prior

\begin{equation}\label{Harm}
Pr(M_{{\bm\gamma}}) =\frac{(1+|{\bm\gamma}|)^{-1}}{{k \choose |{\bm\gamma}|}H_{k+1}}\,,
\end{equation}
where $H_k=\sum_{j=1}^k\, j^{-1}$ represents the $k$-th harmonic number.

This prior is naturally called the {\em Harmonic prior} and assigns 
the collection of models of dimension  $d$ probability proportional to
$(1+d)^{-1}$, which is decreasing almost as slowly as possible in the dimension
and is thus a natural objective choice. 


\subsection{Casella-Moreno-Giron (CMG) prior probabilities}
Borrowing ideas from intrinsic priors, \cite{Cas:Mo:Gi:14} proposed a mixture of a Poisson distribution
\begin{equation}\label{CMG}
Pr(M_{{\bm\gamma}}) = \frac{1}{{k \choose |{\bm\gamma}|}}\int_0^\infty \mbox{Po}(|{\bm\gamma}| \mid\lambda)\, \pi(\lambda) d\lambda \,,
\end{equation}
with $\lambda$ being assigned the density
$$
\pi^I(\lambda)=\frac{e^{-(\lambda+1)}}{\sqrt{\pi\lambda}}\cosh(2\sqrt{\lambda}),\,\, \lambda\ge 0.
$$

\subsection{Half-$k$ Jeffreys prior probabilities}
One simple correction to the Jeffreys model prior probabilities would be to use the Jeffreys probabilities for models up to size $k/2$ if $k$ is even and $(k-1)/2$ if $k$ is odd, and then give the remaining models equal prior probability.  Note that the number
of large models receiving this constant prior probability is


$$ N_k = \left\{
\begin{array}{ll}
\frac{1}{2}\big(2^k-\binom{k}{k/2}\big) & \hbox{if $k$ is even} \\
2^{k-1} & \hbox{if $k$ is odd.}
\end{array}
\right.
$$

The resulting model prior probabilities
unfortunately have an unappealing discontinuity, so we present here a version
which smooths the discontinuity.

In the case of $k$ even:
\begin{equation}\label{Halfk.e}
Pr(M_{\bm\gamma})=\left\{
\begin{array}{ccc}
	C_1^e\frac{1}{(k+1)\binom{k}{|{\bm\gamma}|}} & \mbox{if} & |{\bm\gamma}| \le k/2\\
	C_2^e\frac{k}{2(k+1)N_k} & \mbox{if} & |{\bm\gamma}|\ge k/2+1
\end{array}
\right. \mbox{($k$ even)},
\end{equation}
where $C_1^e$ and $C_2^e$ are such that $Pr(M_{\bm\gamma})$ is a probability mass function and
$$
Pr(M_{\bm\gamma}:|{\bm\gamma}|=k/2)=Pr(M_{{\bm\gamma}'}: |{\bm\gamma}'|=k/2+1).
$$
Solving the corresponding system of equations we arrive at:

\[
C_1^e=
\frac{2(k+1)\binom{k}{k/2}}
{2^k+(k+1)\binom{k}{k/2}},
\qquad
C_2^e=
\frac{2(k+1)\left(2^k-\binom{k}{k/2}\right)}
{k\left(2^k+(k+1)\binom{k}{k/2}\right)}.
\]

Similarly, in the case of $k$ odd:
\begin{equation}\label{Halfk.o}
Pr(M_{\bm\gamma})=\left\{
\begin{array}{ccc}
	C_1^o\frac{1}{(k+1)\binom{k}{|{\bm\gamma}|}} & \mbox{if} & |{\bm\gamma}|\le (k-1)/2\\
	C_2^o\frac{1}{2N_k} & \mbox{if} & |{\bm\gamma}|\ge (k+1)/2
\end{array}
\right.  \mbox{($k$ odd)},
\end{equation}

and imposing
$$
Pr(M_{\bm\gamma}:|{\bm\gamma}|=(k-1)/2)=Pr(M_{{\bm\gamma}'}: |{\bm\gamma}'|=(k+1)/2),
$$
we arrive at:
\[
C_1^o=
\frac{2(k+1)\binom{k}{(k-1)/2}}
{2^k+(k+1)\binom{k}{(k-1)/2}},
\qquad
C_2^o=
\frac{2^{k+1}}
{2^k+(k+1)\binom{k}{(k-1)/2}}.
\]

\subsection{Half-$p$ Jeffreys prior probabilities}
\label{sec.halfp}
The Half-$k$ prior seems sensible and objective but it is rather cumbersome. The analogue of this prior in terms of $p$ --the prior inclusion probability-- is to use a uniform density on half of its parametrix space, that is on $[0,1/2]$.
We actually propose the slight modification
$$
\pi(p)=\mbox{Un}\Big(p\mid 0, \frac{1}{2}\frac{k+1}{k}\Big),
$$
since, for $k=1$ this reduces to the original proposal, and in particular gives 0.5 probability to the null (and to the alternative). We call this prior Half-$p$. 

The prior probability of a model is
\begin{equation}\label{Halfp}
Pr(M_{\bm \gamma}) =\frac{2k}{k+1}\int_0^{\frac{1}{2}\frac{k+1}{k}} p^{|{\bm\gamma}|}
(1-p)^{k-|{\bm\gamma}|} dp=\frac{2k}{k+1}\mbox{Beta}\Big(\frac{1}{2}\frac{k+1}{k}, |{\bm\gamma}|+1, k-|{\bm\gamma}|+1\Big),
\end{equation}
which can be easily computed by means of the incomplete Beta function.

The Half-$p$ is very similar to using $p\sim Un(0,1/2)$, a prior recently proposed by \cite{ViLee20} following a different reasoning (briefly summarized in Section~\ref{sec.toopars}). A small difference is that the Half-$p$ assigns slightly more probability to complex models when $k$ is large, being for instance $e$ times larger for the full model as $k$ tends to infinity.

\subsection{Beta and hierarchical Beta prior probabilities}
An obvious possibility for the inclusion 
probability  $p$ is to use a $Beta(1,\alpha)$, with $\alpha >1$,  and $\alpha=2$ seemingly giving a simple parsimonious choice. The $Beta(1,2)$ choice results in prior model probabilities
\begin{equation}\label{Be1.2}
Pr(M_{\bm \gamma}) =2\int_0^{1} p^{|{\bm\gamma}|}
(1-p)^{1+k-|{\bm\gamma}|} dp=2 \, \mbox{Beta}( |{\bm\gamma}|+1, k-|{\bm\gamma}|+2).
\end{equation}

We also look at a hierarchical prior,
assigning $\alpha$ the {$\pi(\alpha)= \frac{1}{2} \alpha^{-3/2}$} prior density on $(1,\infty)$ (about as flat as possible while maintaining propriety and simplicity). 
The resulting prior for $p$ is
\begin{eqnarray*} \pi^{HB}(p) = \int_1^{\infty} Beta(p \mid 1,\alpha) \frac{1}{2} \alpha^{-3/2} 
d \alpha 
=  \frac
{ \sqrt{\pi} \, \Phi\left(
	-\sqrt{-2\log{(1 - p)}}\right)}{(1 -
	p) \sqrt{(-\log{(1 - p)})}} 
	\end{eqnarray*} and will be called the hierarchical Beta prior.
	The resulting prior probability of a model is
	\begin{equation}\label{HBe}
Pr(M_{\bm \gamma}) =\int_0^1 p^{|{\bm\gamma}|}
(1-p)^{k-|{\bm\gamma}|} \pi^{HB}(p) dp .
\end{equation}
This requires numerical integration, but it only has to be done $k$ times.

\subsection{Overly parsimonious choices from the perspective of objectivity}
\label{sec.toopars}

Several authors have focused on alternative distributions for $p$ in \eqref{General}; \cite{Wils:Iver:Clyd:Etal:10} proposed $\pi(p)=Beta(p\mid 1, k)$ and  \cite{Castillo_etal_15} suggested $\pi(p)=Beta(p\mid a, k^u)$, for certain $a>0$ and $u>1$. 
Since the means of these priors are decreasing at least as fast as $O(1/k)$, they force 
concentration on very small models, which we would not label objective.

Other authors have worked directly with the probability assigned to all models with the same dimension, ${\mathfrak M}_k(d)$. Following this path,  \cite{WoTay-RoFue25} recently introduced the ``Matryoshka doll prior'' that depends on a parameter $\eta$. This prior does not have a closed-form, but the authors prove that ${\mathfrak M}_k(d)$ is given by the Poisson law, Po$(d\mid \log(1+\eta^{-1}))$ as its limiting distribution as $k$ increases. Their default proposals are $\eta=1$ and $\eta=1/(e-1)$. Thus ${\mathfrak M}_k(d)$ behaves like $1/d!$, which
is too sharply decreasing in model size to be labeled objective.

Finally, \cite{ViLee20} proposed $Pr(M_{\bm\gamma})\propto e^{-c|{\bm\gamma}|}$ with $c=1.2785$, but acknowledge that it does not provide needed parsimony. To fix this issue they suggest using $c=\log(k)$ or to use a hyperprior on $c$. Following this last path, and after an interesting discussion, the authors arrive at a distribution on $c$ which is equivalent to using $\pi(p)=\mbox{Un}(p\mid 0,1/2)$ that we already mentioned in a previous section.

\section{Basic comparisons} 
With the  model prior probabilities described in the previous section plus the Uniform and Jeffreys', we have a total of 8 different model prior probabilities. These are listed in Table~\ref{priors} and their corresponding functions ${\mathfrak M}_k(d)$ are collected in Table~\ref{expres_probs_model_size}, along with their expectations that will be needed later.	 We begin their comparison by looking at some essential properties.

\subsection{Philosophical issues}

\subsubsection{When $k=1$}
If $k=1$ then it seems natural to require $Pr(M_0)=Pr(M_1)=1/2$, the
default choice for hypothesis testing. The Uniform, Jeffreys, Half-$p$ and Half-$k$ 
prior probabilities satisfy this. The rest do assign significant mass to the null model: CMG, Harmonic and Be(1,2) assigns 2/3 to $M_0$ and the Hierarchical 0.78. 


\subsubsection{Simplicity}
All the model prior probabilities are closed form, except the Hierarchical and CMG priors
which require evaluation of $k$ one-dimensional integrals. The Half-$k$ prior is rather
cumbersome and is thus somewhat less appealing.

\subsection{Graphs of the priors}
Figure \ref{fig.allpriors} contains two graphs
of the prior model probabilities for a total number of covariates $k=47$ (the number of explanatory variables in the infant obesity dataset of Section~\ref{sec.obesity}). The left graph gives the log of the model prior probability
for a model of the indicated size $d$ while the right graph gives $\mathfrak{M}_k(d)$.

In this graph, we observe that the CMG prior kills off all models of dimension bigger than 14; the Half-$k$ and Half-$p$ priors kill off models of size bigger than 34, while the hierarchical, Beta(1,2) and Harmonic choices do give non-negligible mass to all model sizes. Interestingly, all the prior probability choices (except the Uniform) clearly give more mass to smaller models than
does the Jeffreys prior. 
We will use these plots to illustrate the general arguments in the next sections.

\begin{table}[t!]
\begin{tabular}{lc|lc}
Name & $Pr(M_{\bm\gamma})$ & Name & $Pr(M_{\bm\gamma})$\\
\hline
Uniform & \eqref{Uniform}&Jeffreys & \eqref{Jeff}\\
Harmonic & \eqref{Harm}&CMG & \eqref{CMG}\\
Half-$k$ & \eqref{Halfk.e},\eqref{Halfk.o}&Half-$p$ & \eqref{Halfp}\\
Beta(1,2) & \eqref{Be1.2}&Hierarchical beta & \eqref{HBe}\\
\hline
\end{tabular}
\caption{For the eight priors considered: names and the equations where they are defined.}
\label{priors}
\end{table}

\begin{table}[t!]
\centering
\begin{adjustbox}{max width=\textwidth}
\begin{tabular}{l|ll}
	Prior & $\mathfrak{M}_k(d)$ & $E^{\mathfrak{M}_k}(d)$\\[.1cm] \hline
	Uniform 
	&  $\binom{k}{d}/2^k$ & $\frac{k}{2}$ \\[.1cm]
	Jeffreys 
	& $1/(k+1)$ & $\frac{k}{2}$
	\\[.1cm]
	Beta$(1,2)$
	& $\frac{2}{k+2}(1-\frac{d}{k+1})$ & $\frac{k}{3}$\\[.1cm]
	Half-$p$
	& $ {k \choose d} \frac{2k}{k+1}\mbox{Beta}\Big(\frac{1}{2}\frac{k+1}{k}, d+1, k-d+1\Big)$ & $k\frac{k+1}{4k}\approx\frac{k}{4}$\\[.1cm]
	Hierarchical
	&  ${k \choose d}\int_0^1 p^{d}
	(1-p)^{k-d} \pi^{HB}(p) dp$ & 
	$k\int_0^1\, p\, \pi^{HB}(p)dp$
	\\[.1cm]
	CMG
	& $\frac{ \int_0^{\infty} Po(d \mid \lambda) \pi^I(\lambda) d \lambda} {  \sum_{d=0}^k \int_0^{\infty} Po(d \mid \lambda) \pi^I(\lambda) d \lambda}$ & 
	$\frac{ \int_0^{\infty} Po(d \mid \lambda) \pi^I(\lambda) d \lambda} {\sum_{d=0}^k \int_0^{\infty} Po(d \mid \lambda) \pi^I(\lambda) d \lambda}
	\approx \int_0^\infty\,\lambda\,\pi^I(\lambda)d\lambda$ 			
	\\[.1cm]
	Harmonic
	&  $\frac{(1+d)^{-1}}{H_{k+1}}$ & 
	$\frac{k+1}{H_{k+1}}-1\approx \frac{k}{\log(k)}$
	\\[.1cm]
	Half-$k$ ($k$ even) &
	$\Big\{ \begin{array}{ccc}
		C_1^e\frac{1}{(k+1)} & \mbox{if} & d \le k/2\\
		C_2^e\binom{k}{d}\frac{k}{2(k+1)N_k} & \mbox{if} & d\ge k/2+1
	\end{array}$ 
	&
	$k\frac{2^{k-1}+\frac{k+1}{4}\binom{k}{k/2}}{2^k+(k+1)\binom{k}{k/2}}\approx\frac{k}{4}$\\
	Half-$k$ ($k$ odd) &
	$\Big\{ \begin{array}{ccc}
		C_1^o\frac{1}{(k+1)} & \mbox{if} & d\le (k-1)/2\\
		C_2^o\binom{k}{d}\frac{1}{2N_k} & \mbox{if} & d\ge (k+1)/2
	\end{array}
	$ & 
	$k\frac{2^{k-1}+\frac{k^2-1}{4k}\binom{k}{(k-1)/2}}{2^k+(k+1)\binom{k}{(k-1)/2}}\approx\frac{k}{4}$
	\\
	\hline
\end{tabular}
\end{adjustbox}
\caption{Expressions for $\mathfrak{M}_k(d)$ and the expectation of $d\sim \mathfrak{M}_k(d)$ for the various priors. $\approx$ indicates the first order term in the large $k$ approximation to the expectation. The computation of $E^{\mathfrak{M}_k}(d)$ for the mixture over $p$ priors follows from the straightforward identity
$E^{\mathfrak{M}_k}(d) = k E^{\pi}(p)$. Derivation of the large $k$ approximations is straightforward, often using Stirling's approximation. Numerical integration yields $\int_0^\infty\,\lambda\,\pi^I(\lambda)d\lambda \approx 1.5$ and $\int_0^1 p\,\pi^{HB}(p)dp \approx 0.22$.}
\label{expres_probs_model_size}	
\end{table}

\section{Two notions of parsimony}

\subsection{Weak parsimony}

It is natural to require that 
\begin{equation}
\label{eq.WP}
{\mathfrak M}_k(d+1)\le {\mathfrak M}_k(d) \quad \mbox{ for all} \quad d \,,
\end{equation}
so that the sum of model prior probabilities of a given size is non-increasing in size. We call this {\em weak parsimony (WP)}.

The right graph in Figure~\ref{fig.allpriors} illustrates this, with 
 all priors considered (except the Uniform) satisfying WP. Inspection of Table~\ref{expres_probs_model_size} confirms that the Jeffreys, Beta(1,2), Harmonic and Half-$k$ priors always satisfy WP. For priors that are obtained as mixtures over $p$ or over the Poisson mass function, the next two results show that non-increasing mixture functions produce priors that are WP. This confirms that the CMG, Half-$p$ and Hierarchical prior probabilities are WP.

\begin{lemma}\label{WP.pip}
Weak parsimony (WP) holds for priors of the form \eqref{General} with $\pi(p)$ non-increasing on \([0,1]\).
\end{lemma}

\begin{proof}
In this case \eqref{eq.WP} reduces to showing that 
\[
\binom{k}{d+1}
\int_0^1 p^{d+1}(1-p)^{k-d-1}\pi(p)\,dp
\le
\binom{k}{d}
\int_0^1 p^d(1-p)^{k-d}\pi(p)\,dp,
\]
for all $d\in\{0,\ldots,k-1\}$. 
After straightforward algebra we obtain that this is equivalent to
\begin{equation}\label{eq.ineq1}
\int_0^1
p^d(1-p)^{k-d-1}\, p\, \pi(p)\, dp\le \frac{d+1}{k+1}\, \int_0^1\, p^d(1-p)^{k-d-1}\, \pi(p)\, dp.
\end{equation}
To show that this inequality holds true, under the conditions of the result, we use Chebyshev's integral inequality, which states that, for $f,g$ integrable functions with one non-increasing and the other increasing, and for a probability density function  $\mu$, 
\begin{equation}\label{eq.Chev}
\int f(p)\, g(p)\, \mu(p)\, dp\le \int f(p)\,\mu(p)\, dp\; \int g(p)\,\mu(p)\, dp.
\end{equation}
Choosing $f(p)=p$ (increasing), $g(p)=\pi(p)$ (non-increasing by hypothesis) and $\mu(p)=Be(p\mid d+1,k-d)$ results in the inequality \eqref{eq.ineq1}.
\end{proof}

\begin{lemma}
Weak parsimony (WP) holds for priors of the form \eqref{CMG}, with $\pi(\lambda)$ non-increasing on \([0,\infty)\).
\end{lemma}
\begin{proof}
In this case \eqref{eq.WP} reduces to showing that 
$$
\int_0^\infty\, \frac{e^{-\lambda}\lambda^{d+1}}{(d+1)!}\pi(\lambda)d\lambda\le
\int_0^\infty\, \frac{e^{-\lambda}\lambda^{d}}{d!}\pi(\lambda)d\lambda,
$$
for all $d\in\{0,\ldots,k-1\}$,
or equivalently
$$
\int_0^\infty\, e^{-\lambda}\lambda^{d}\,\lambda\, \pi(\lambda)d\lambda\le (d+1)\int_0^\infty\, e^{-\lambda}\lambda^{d}\,\pi(\lambda)d\lambda.
$$
Using Chebyshev's integral inequality \eqref{eq.Chev}, with $f(\lambda)=\lambda$ (increasing), $g(\lambda)=\pi(\lambda)$ (non-increasing by hypothesis), and $\mu(\lambda)=Gamma(\lambda\mid d+1,1)$ establishes the result.

\end{proof}



\subsection{Strong parsimony}
We define {\em strong parsimony (SP)} as the requirement that individual model prior probabilities should be non-increasing as model size grows, that is
\begin{equation}
\label{eq.SP}
Pr(M_{\bm\gamma:|\bm\gamma|=d+1})\le Pr(M_{\bm\gamma:|\bm\gamma|=d}) \quad \mbox{ for all} \quad d \,.
\end{equation}
The left plot in Figure~\ref{fig.allpriors} illustrates this for the case $k=47$. The Uniform prior obviously satisfies SP (since it is constant) and, among the rest, only the Half-$k$ and Half-$p$ prior probabilities are candidates to satisfy SP. Looking at the expressions \eqref{Halfk.e} and \eqref{Halfk.o} we straightforwardly conclude that Half-$k$ satisfies SP. The result for the Half-$p$ is more involved. Obviously, mixtures of the type \eqref{General} with non-increasing $\pi(p)$ do not necesarily guarantees SP, since the Jeffreys prior violates strong parsimony. The following result establishes a characterization of SP for this class of priors.

\begin{lemma} 
Strong parsimony (SP) holds for priors of the form \eqref{General}	 if and only if
\begin{equation}\label{eq.SPH}
\int_0^1 p^{k-1}(1-2p)\pi(p)\,dp\ge 0.
\end{equation}

\end{lemma}

\begin{proof}
Notice that
\[
Pr(M_{\bm\gamma:|\bm\gamma|=d})
=
\int_0^1
\left(\frac{p}{1-p}\right)^d
(1-p)^k\pi(p)\,dp.	
\]
Applying the Cauchy-Schwarz inequality yields
\begin{eqnarray*}
Pr(M_{\bm\gamma:|\bm\gamma|=d})^2
&=&
\Big(\int_0^1
\left(\frac{p}{1-p}\right)^d
(1-p)^k\pi(p)\,dp
\Big)^2\\
&=&
\Big(\int_0^1
\left(\frac{p}{1-p}\right)^{(d+1)/2}(1-p)^{k/2}
\left(\frac{p}{1-p}\right)^{(d-1)/2}
(1-p)^{k/2}\pi(p)\,dp
\Big)^2\\
&\le&
\int_0^1
\left(\frac{p}{1-p}\right)^{d+1}
(1-p)^k\pi(p)\,dp\;
\int_0^1
\left(\frac{p}{1-p}\right)^{d-1}
(1-p)^k\pi(p)\,dp\\
&=& Pr(M_{\bm\gamma:|\bm\gamma|=d+1})Pr(M_{\bm\gamma:|\bm\gamma|=d-1}).
\end{eqnarray*}
Equivalently, consecutive ratios are non-decreasing
$$
\frac{Pr(M_{\bm\gamma:|\bm\gamma|=1})}{Pr(M_{\bm\gamma:|\bm\gamma|=0})}\le 
\frac{Pr(M_{\bm\gamma:|\bm\gamma|=2})}{Pr(M_{\bm\gamma:|\bm\gamma|=1})}\le
\cdots
\le\frac{Pr(M_{\bm\gamma:|\bm\gamma|=k})}{Pr(M_{\bm\gamma:|\bm\gamma|=k-1})}.
$$
The above inequalities allows to assert that SP holds if and only if the largest consecutive ratio is at most one, namely
\[
\frac{Pr(M_{\bm\gamma:|\bm\gamma|=k})}{Pr(M_{\bm\gamma:|\bm\gamma|=k-1})}\le 1.
\]
Algebra shows that this is equivalent to \eqref{eq.SPH}.
\end{proof}

A consequence of the above lemma is that the Half-$p$ satisfies SP since 
$$
\int_0^1\, p^{k-1}(1-2p)\, \mbox{Un}(p\mid 0, \frac{1}{2}\frac{k+1}{k})=0.
$$
For  $p\sim\mbox{Be}(a,b)$, \ref{eq.SPH} can be shown to hold only if $b-a\ge k-1$.
However, these priors do not satisfy WP and, for increasing $a$ and $b$, lead to degenerate priors (eg. $a=k, b=2k-1$ converges
to a point mass at  $p=1/2$ leading to the Uniform assignment of probabilities).

An interesting observation arises when we compare both types of parsimony. 
It can be easily shown that, if SP holds, then \eqref{eq.WP} is satisfied for $d\in\{0,\ldots,k/2-1\}$ if $k$ is even and for $d\in\{0,\ldots,(k-1)/2\}$ if $k$ is odd. The following result focuses on the remaining model sizes, showing that imposing non-decreasing individual model prior probabilities forces the collective probabilities to decrease at a fast rate.

\begin{lemma}\label{lemmak.2}
Suppose $k$ is even and \eqref{eq.SP} holds for all $d\in\{k/2,\ldots,k-1\}$. Then
\begin{equation}\label{eq.nec}
{\mathfrak M}_k(d+1)\le {\mathfrak M}_k\Big(\frac{k}{2}\Big)\frac{\binom{k/2}{d+1-k/2}}{\binom{d+1}{d+1-k/2}},
\end{equation}
for all $d\in\{k/2,\ldots,k-1\}$.

If $k$ is odd and \eqref{eq.SP} holds for all $d\in\{(k+1)/2,\ldots,k-1\}$, then
$$
{\mathfrak M}_k(d+1)\le {\mathfrak M}_k\Big(\frac{k+1}{2}\Big)\frac{\binom{(k-1)/2}{d-(k-1)/2}}{\binom{d+1}{d-(k-1)/2}},
$$
for all $d\in\{(k-1)/2,\ldots,k-1\}$.

\end{lemma}
\begin{proof} 
The proof follows by induction. Suppose $k$ is even. Under the conditions of the lemma, first notice that \eqref{eq.nec} is trivially true for $d=k/2$. Supposing that it is satisfied for some  $d$ in the interval $[k/2+1, k-2]$, we now show that it is also satisfied for $d+1$:
\begin{eqnarray*}
{\mathfrak M}_k(d+2)&=&Pr(M_{\bm\gamma:|\bm\gamma|=d+2})\binom{k}{d+2}\le
Pr(M_{\bm\gamma:|\bm\gamma|=d+1})\binom{k}{d+2}=\frac{{\mathfrak M}_k(d+1)}{\binom{k}{d+1}}\binom{k}{d+2}\\
&\le& {\mathfrak M}_k\Big(\frac{k}{2}\Big)\frac{\binom{k/2}{d+1-k/2}}{\binom{d+1}{d+1-k/2}}\frac{\binom{k}{d+2}}{\binom{k}{d+1}}=
{\mathfrak M}_k\Big(\frac{k}{2}\Big) \frac{\binom{k/2}{d+2-k/2}}{\binom{d+2}{d+2-k/2}}.
\end{eqnarray*}
The proof for $k$ odd is similar and is omitted.
\end{proof}
Based on the result above, a necessary condition for SP is ${\mathfrak M}_k(k)\le {\mathfrak M}_k(k/2)/\binom{k}{k/2}$ and, since ${\mathfrak M}_k(k/2)\le 1$, we have ${\mathfrak M}_k(k)\le 1/\binom{k}{k/2} = O(2^{-k}\sqrt{k})$.  
Since ${\mathfrak M}_k(k)=Pr(M_k)$, this means that the probability of the full model can be at most $O(2^{-k}\sqrt{k}$)), a
very small number. 
Returning to the illustrative case with $k=47$, that the Half-$k$ and Half-$p$ priors essentially kill off all models of size greater than 34  is in fact a consequence of SP.

\begin{figure}[t!]
\begin{center}
\includegraphics[height=7cm, width=10cm]{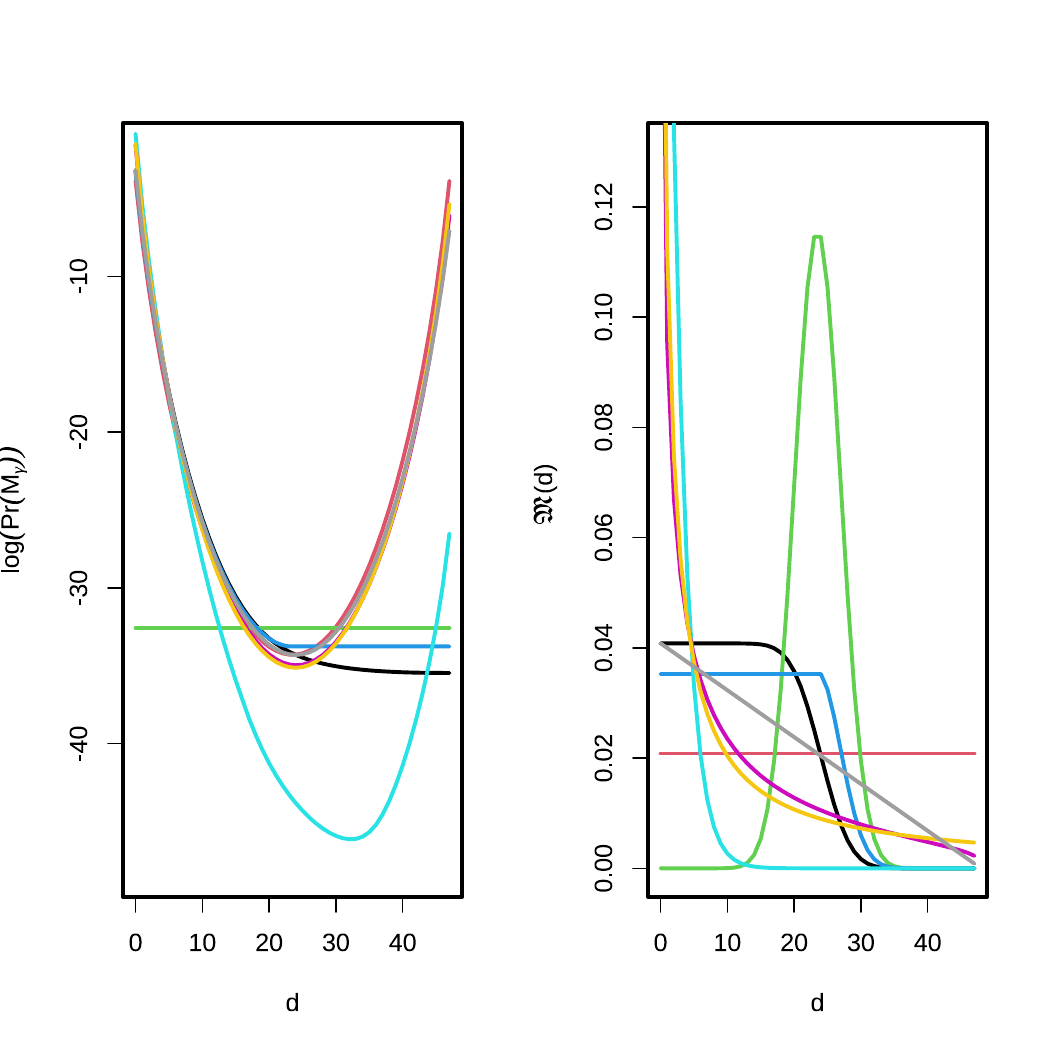}
\end{center}
\caption{For $k=47$, left: $\log(Pr(M_{\bm\gamma}))$ with $|\bm\gamma|=d$; right: ${\mathfrak M}_k(d)$. Priors: Half-$p$ (black); Jeffreys (red), Uniform (Green), Half-$k$ (blue); CMG (Cyan); Hierarchical (magenta); Beta(1,2) (gray); Harmonic (yellow).}
\label{fig.allpriors}
\end{figure}

\section{Large $k$ behaviour and the control for multiplicity}
In this section we focus on the behaviour of model priors in particular model sizes:  $d=0,1$ (small models), $d=k/2$ (medium sized
model) and $d=k$ (full model), to understand extreme behaviors. Table \ref{probs_model_size} gives  values of $\mathfrak{M}_k(d)$ at these values of $d$ for the various prior
probabilities. Of particularly interest is the situation where $k$ grows. We have included large $k$ approximations to $\mathfrak{M}_k(d)$. The validity of these approximations is proved in the following result.


\begin{table}[t!]
\centering
\begin{tabular}{c|cccc}
$d$ & Uniform & Jeffreys & Beta(1,2) & Half-$p$ \\[.2mm]
\hline
$0$ 
& $=2^{-k}$ 
& $=\frac{1}{k+1}$ 
& $=\frac{2}{k+2}$ 
& $\approx \frac{2}{k}$ \\[.1cm] 

$1$ 
& $=k 2^{-k}$ 
& $=\frac{1}{k+1}$  
& $=\frac{2k}{(k+1)(k+2)}$ 
& $\approx \frac{2}{k}$ \\[.1cm]

$k/2$ 
& $\approx \sqrt{\frac{2}{k \pi}}$  
& $=\frac{1}{k+1}$  
& $=\frac{1}{k+1}$ 
& $\approx \frac{1}{k}$ \\[.1cm] 

$k$ 
& $=2^{-k}$ 
& $=\frac{1}{k+1}$  
& $=\frac{2}{(k+1)(k+2)}$ 
& $\approx \frac{e}{k2^k}$ \\[.2cm] 
\hline

$d$ & Hierarchical & CMG & Harmonic & Half-$k$\\[.1cm]
\hline
$0$ 
& $\approx \frac{\pi}{2\sqrt{k}}$ 
& $\approx 0.43$ 
& $\approx\frac{1}{\log(k+1)}$ 
& $\approx\frac{2}{k}$\\[.1cm] 

$1$ 
& $\approx \frac{\pi}{4\sqrt{k}}$ 
& $\approx 0.21$ 
& $\approx\frac{1}{2\,\log(k+1)}$ 
& $\approx\frac{2}{k}$\\[.1cm] 

$k/2$ 
& $\approx \frac{0.51}{k}$ 
& $< O\left(\frac{ke^{\sqrt{k}}}{2^{k/2}}\right)$
& $\approx \frac{2}{k\log k}$ 
& $\approx\frac{2}{k}$ \\[.1cm] 

$k$ 
& $\approx \frac{1}{2k\log k}$ 
& $< O\left(\frac{ke^{\sqrt{2k}}}{2^k}\right)$
& $\approx \frac{1}{k\log k}$ 
& $\approx \sqrt{\frac{2\pi}{k}}\frac{1}{2^k}$ \\[.1cm]
\hline
\end{tabular}
\caption{Values of $\mathfrak{M}_k(d)$ for the various priors and specified values of $d$. \\
$\approx$ indicates that the ratio of the table entry and $\mathfrak{M}_k(d)$ goes to 0 as $k$ grows. }
\label{probs_model_size}	
\end{table}

\begin{lemma}
The values of $\mathfrak{M}_k(d)$ in Table~\ref{probs_model_size} are, as indicated, exact, given an
upper bound, or approximated ($\approx$) in the sense the ratio of the table entry and $\mathfrak{M}_k(d)$
goest to zero as $k$ increases.
\end{lemma}

\begin{proof}\phantom{h}
For Jeffreys' prior, $\mathfrak{M}_k(d)=1/(k+1)$ for all $d$.
For the Uniform prior, $\mathfrak{M}_k(d)=\binom{k}{d}/2^k$ and Stirling's approximation yields the result for $d=k/2$. 

For the Half-$p$,
$$
\mathfrak{M}_k(0)=\frac{2k}{k+1}\, Beta(\frac{1}{2}\frac{k+1}{k},1,k+1)\approx \frac{2}{k},
$$
where the approximation of the integral follows from explicitly evaluating the integral 
corresponding to the incomplete Beta function and then using standard limiting arguments.
Similarly
\begin{eqnarray*}
\mathfrak{M}_k(1)&=&k\,\frac{2k}{k+1}\, Beta(\frac{1}{2}\frac{k+1}{k},2,k)\\
&=&\frac{2k^2}{k+1}\,\Big(\frac{1}{k(k+1)}
-\frac{1}{k}\left(\frac{k-1}{2k}\right)^k\
+\frac{1}{k+1}\left(\frac{k-1}{2k}\right)^{k+1}\Big)\\ 
&\approx& \frac{2k^2}{k+1}\frac{1}{k(k+1)}\approx \frac{2}{k} \,.
\end{eqnarray*}
For $d=k/2$, 
\begin{eqnarray*}
\mathfrak{M}_k(k/2)&=&\binom{k}{k/2} \,\frac{2k}{k+1}\, \int_0^{\frac{k+1}{2k}} \big(p(1-p)\big)^{k/2}dp\\
&=&\binom{k}{k/2} \,\frac{2k}{k+1}\, \Big(\int_0^{1/2} \big(p(1-p)\big)^{k/2}dp+\int_{1/2}^{\frac{k+1}{2k}} \big(p(1-p)\big)^{k/2}dp\Big)\\
\end{eqnarray*}
The integrand is symmetric with respect to $p=1/2$ and decreasing with $p\in[1/2,1]$ and hence
$$
\int_0^{1/2} \big(p(1-p)\big)^{k/2}dp=\frac{1}{2}\mbox{Beta}(k/2+1,k/2+1)\approx \sqrt{\frac{2\pi}{k+2}}2^{-k-2}
$$
(applying Stirling's approximation), and
$$
\int_{1/2}^{\frac{k+1}{2k}} \big(p(1-p)\big)^{k/2}dp\le \Big(\frac{1}{4}\Big)^{k/2}\frac{1}{2}\Big(\frac{k+1}{k}-1\Big).
$$
It follows that
\begin{eqnarray*}
\mathfrak{M}_k(k/2) &=& \binom{k}{k/2} \,\frac{2k}{k+1}\Big(\frac{1}{2}\mbox{Beta}(k/2+1,k/2+1)+O\Big(\frac{2^{-k}}{k}\Big) \Big)\\
&\approx&  \binom{k}{k/2} \,\frac{2k}{k+1}\Big(\sqrt{\frac{2\pi}{k+2}}2^{-k-2}+O\Big(\frac{2^{-k}}{k}\Big) \Big)\\
&\approx&  \binom{k}{k/2} \,\frac{2k}{k+1}\Big(\sqrt{\frac{2\pi}{k+2}}2^{-k-2}\Big)\approx \frac{1}{k},
\end{eqnarray*}	
since $\binom{k}{k/2}\approx \sqrt{\frac{2}{\pi k}}2^{k}$ by Stilrling's approximation. Finally, the approximation for $\mathfrak{M}_k(k)$ follows by exact integration and standard asymptotic results.

To obtain the results concerning the Hierarchical Beta, we start by deriving a useful result for the mixture over $p$ priors and if $k$ and $d$ are large, 
\begin{eqnarray} \mathfrak{M}_k(d) &=& \left(
\begin{array}{c}
	k \\
	d \\
\end{array}
\right)
\int_0^1 p^{d} (1-p)^{k-d} \pi(p)dp   \nonumber\\
&=& \frac{1}{k+1} \int_0^1 Beta(p \mid d+1, k-d+1) \pi(p) dp 
\approx \frac{\pi(d/k)}{k} \label{eqMd}
\end{eqnarray}
if $\pi(v)$ is continuous and nonzero. This follows because the mean of
the Beta distribution is approximately $\mu = \frac{d}{k}$ and the standard deviation
is smaller than $\frac{\mu}{\sqrt d}$, so the Beta distribution is concentrated around the mean. Hence $\mathfrak{M}_k(d) ={k \choose d}\int_0^1 p^{d}
(1-p)^{k-d} \pi^{HB}(p) dp \approx \frac{\pi^{HB}(\frac{d}{k})}{k}$ for large $d$ and
$d/k$ bounded away from zero (necessary because $\pi^{HB}$ is infinite at zero).
Note that, as $\frac{d}{k} \rightarrow 1$, $\pi^{HB}(\frac{d}{k}) \approx \frac{1}{-2\log(1-\frac{d}{k})}$. Similarly, $\mathfrak{M}_k(k/2)\approx \frac{\pi^{HB}(0.5)}{k}\approx \frac{0.51}{k}$.

The results for the Harmonic follows from $\mathfrak{M}_k(d) = \frac{(1+d)^{-1}}{H_{k+1}} \approx \frac{(1+d)^{-1}}{\log(k+1)} $ if $k$ is large.

The approximations for the Half-$k$ are obtained noticing that
$C_1^e\approx 2$ and $C_2^e\approx \frac{\sqrt{2\pi}}{k},$
and using similar standard large $k$ approximations to the remaining expressions.

The analysis for the CMG prior, when $d=k$, begins with the bound
$$  \int_0^{\infty} Po(k \mid \lambda) \pi^I(\lambda) d \lambda < 
\frac{1}{k!}  \int_0^{\infty} \lambda^{k-0.5} e^{-2(\lambda - \sqrt{\lambda})} d \lambda \,,
$$
which follows from $\cosh(2\sqrt{\lambda}) < e^{2\sqrt{\lambda}}$. Making the change of variables $\lambda = x^2$ results in
the expression 
$$ \frac{1}{k!} \int_0^{\infty} x^{2k} e^{-2(x^2 -x)}dx <  \frac{1}{k!} \int_{-\infty}^{\infty} x^{2k} e^{-2(x- 0.5)^2}dx 
=  \frac{1}{k!} \sum_{j=0}^{k} \left(
                                   \begin{array}{c}
                                     2k \\
                                     2j \\ 
                                   \end{array}
                                 \right) \frac{(2j)!}{2^{2k+j} j!} \,,
$$
the last step following from the expression for the $2k$ moment of a normal distribution with mean 1/2 and variance 1/4.
Using Stirling's approximation and dropping a term $\sqrt{k-j}$ in the denominators (the results being larger) yields a summation whose terms can be maximized over $j$, the maximum term being 
$O(e^{\sqrt{2k}}/2^k)$. Since there are only $k$ terms in the summation and the denominator for the expression for the
CMG in Table~\label{expres_probs_model_size} can be shown to converge to 1 as $k$ increases, the result follows.
The same argument works for $d=k/2$.

\end{proof}

\subsubsection{Full and null model}
Looking at Table~\ref{probs_model_size}, it can be argued that the Uniform, Half-$p$, and Half-$k$ priors give too little probability to the full model. In fact, these are more conservative (in support of the full model) than the bound $O(2^{-k}\sqrt{k})$ that we showed is a necessary condition for SP (Half-$p$ being the most conservative, then Half-$k$ and then the Uniform). The CMG prior can also be argued to give too little probability to the full model, even though the table entry does not
violate the SP condition; this is almost certainly due to the entry being a  too large upper bound, since
the numerical results all suggested that the CMG prior penalizes the full model far more than the others. The Hierarchical and Harmonic priors give essentially the same weight to the full model and are slightly less conservative than the Beta(1,2). 

For the null model and the models with just one covariate, the CMG prior gives by far the most weight to small
models. the Harmonic prior and then the Hierarchical Beta prior are next
in terms of favoring small models.  Jeffreys', Beta(1,2), Half-$p$ and Half-$k$ behave similarly in support for small models and the Uniform provides the least weight to small models.

\subsubsection{Midsize models}
The most populated model size is $d=k/2$ and, hence, $\mathfrak{M}_k(k/2)$ is intuitively an indicator of multiplicity adjustment; if not controlled by multiplicity adjustment, the midsize models will dominate in terms of prior probabilities (and often posterior probabilities).
It was shown in \cite{Scott:Berg:10} that the Uniform (with $\mathfrak{M}_k(k/2)\approx (2/k\pi)^{1/2}$) does not correct for multiplicity, while the Jeffreys' prior, with $\mathfrak{M}_k(k/2)=1/(k+1)$ does. 
Taking the Jeffreys prior as a standard for multiplicity control, we can evaluate the degree of multiplicity control provided by the rest of proposals: priors with $\mathfrak{M}_k(k/2)\approx \frac{C}{k}$ provide a Jeffreys' type adjustment for multiplicity. This is the case for all the mixture over $p$ priors and the Half-$k$. The Harmonic prior has an additional $1/ \log k$ penalty while the CMG probabilities are essentially zero, killing all models of midsize.

\paragraph{Combining parsimony and control for multiplicity} The previous notion of Jeffreys' adjustment for multiplicity combined with Strict Parsimony (SP) provides an interesting characterization of model prior probabilities. According to Lemma~\ref{lemmak.2}, for priors with ${\mathfrak M}_k(k/2)\approx C\frac{1}{k}$ a necessary condition for SP becomes ${\mathfrak M}_k(k)=O(2^{-k}/\sqrt{k})$. Thus a prior satisfying ${\mathfrak M}_k(k)\approx C 2^{-k}/\sqrt{k}$ would be  the least parsimonous among the class of priors that satisfy SP and that provide a Jeffreys' type adjustment for multiplicity. Among the priors studied here, only the Half-$k$ has this property.



\section{Prior variable inclusion probabilities}
The prior inclusion probability for a variable $x_j$, $j\in\{1,\ldots,k\}$ is 
$$
Pr(x_j):= Pr(\cup\{M_{\bm\gamma}:\gamma_j=1\}) \,,
$$
i.e., the overall prior probability that the given variable occurs in models.
The following result relates these probabilities to the expectation of $d\sim\mathfrak{M}_k(d)$. 
\begin{lemma}
For priors satisfying \eqref{eq.Mg.dim}, $Pr(x_j)=E^{\mathfrak{M}_k}(d)/k$.
\end{lemma}
\begin{proof}
Condition on the set of models of size $d$. Without loss of generality choose a model by sequential selection of variables. Because of exchangeability, the probability
that variable $x_j $ is chosen first is $1/k$. Indeed, the probability that $x_j$ is chosen in any position is $1/k$. Since these are mutually exclusive events,
the overall conditional inclusion probability of $x_j$ is $d/k$. Probabilistically averaging over $d$ gives the result.

\end{proof}
Table~\ref{expres_probs_model_size} presented  $E^{\mathfrak{M}_k}(d)$ for the priors, together with their large $k$ approximations.
For the mixture over $p$ priors and since $E^{\mathfrak{M}_k}(d)=k\, E^{\pi}(p)$, it is immediate that $Pr(x_j)= E^{\pi}(p)$.   Table~\ref{tab.Pxj}  gives the  prior inclusion probabilities for various values of $k$.

\begin{table}[t!]
\centering
\begin{adjustbox}{max width=\textwidth}
\begin{tabular}{c|ccccccccc}
	$k$ & Uniform & Jeffreys & Beta(1,2)&Half-$p$ &Hierarchical&CMG& Harm & Half-$k$ \\ \hline
	$1$ & 0.50 & 0.50 & 0.33 & 0.50 &0.22 & 0.33 & 0.33  &  0.50 \\ 
	$3$ & 0.50 & 0.50 & 0.33& 0.33 &0.22 & 0.57 & 0.31& 0.40\\ 
	$5$ & 0.50 & 0.50 &0.33 & 0.30 & 0.22& 0.43 & 0.29 & 0.37\\ 
	$7$ & 0.50 & 0.50 &0.33 &0.29 &0.22 & 0.34 & 0.28 & 0.35\\ 
	$9$ & 0.50 & 0.50 & 0.33&0.28 &0.22 & 0.28 & 0.27& 0.34 \\ 
	$20$ & 0.50 & 0.50 &0.33 &0.26 &0.22 & 0.13 & 0.24 & 0.31\\ 
	$200$ & 0.50 & 0.50 &0.33 &0.25 & 0.22& 0.01 &0.17 & 0.27\\ \hline
	$\infty$ & 0.50 & 0.50 & 0.33 & 0.25 & 0.22 & 0 & 0 & 0.25
\end{tabular} %
\end{adjustbox}
\caption{Prior variable inclusion probabilties for the various priors}
\label{tab.Pxj}
\end{table}

Initial  thoughts were that objectivity would imply that the prior variable inclusion probability should be 0.5. We label this property as Hard objectivity (HO). Of the priors studied, this is only satisfied by the Uniform and 
Jeffreys priors, with the other priors having $Pr(x_j)$ strictly smaller than 0.5.  Thus having a prior inclusion probability of 0.5 seems to be incompatible with weak parsimony (WP). In Lemma~\ref{WP.pip}, we showed that $\pi(p)$ being non-increasing guarantees WP but these priors have an expectation bounded above by 0.5 (the $p\sim \mbox{Unif}(0,1)$ being the limiting case).

Slightly less demanding is having a strictly positive limiting value for $Pr(x_j)$, as $k$ grows. We name this characteristic  Soft objectivity (SO). This property is satisfied by all priors studied except the CMG and Harmonic priors. 

Our view is that WP is considerably more important than HO, so we would still rule out the uniform and Jeffreys prior probabilities. The situation with SO is not so clear; arguments can be made both ways.
 If one is including only variables
that are viewed as being of legitimate scientific interest, then adding more
legitimate variables should not cause the prior variable inclusion probabilities
to decrease. But, if one is including lots of `junk' variables, then it would be
natural to have the prior inclusion probabilities decrease.

\section{Conclusions}

\begin{table}[t!]
\centering
\begin{adjustbox}{max width=\textwidth}
\begin{tabular}{lcccccccc}
	& Uniform & Jeffreys & Beta(1,2) & Hierarchical & CMG & Harmonic & Half-$p$ & Half-$k$ \\ \hline
	\multicolumn{8}{l}{\emph{Philosophical considerations}}\\
	1. Hypothesis testing & \ding{52} & \ding{52} & \ding{56} & \ding{56} & \ding{56} & \ding{56} & \ding{52} & \ding{52} \\
	2. Simplicity & \ding{52} & \ding{52} & \ding{52} & \ding{56} & \ding{56} & \ding{52} & \ding{52} & \ding{56} \\
	\multicolumn{8}{l}{\emph{Parsimony}}\\
	3. Weak (WP) & \ding{56} & \ding{52} & \ding{52} & \ding{52} & \ding{52} & \ding{52} & \ding{52} & \ding{52} \\
	4. Strong (SP) & \ding{52} & \ding{56} & \ding{56} & \ding{56} & \ding{56} & \ding{56} & \ding{52} & \ding{52} \\
	\multicolumn{8}{l}{\emph{Multiplicity}}\\
	5. Control & \ding{56} & \ding{52} & \ding{52} & \ding{52} & \ding{52} & \ding{52} & \ding{52} & \ding{52} \\
	6. Jeffreys' Control & \ding{56} & \ding{52} & \ding{52} & \ding{52} & \ding{56} & \ding{56} & \ding{52} & \ding{52} \\
	\multicolumn{8}{l}{\emph{Objectivity in inclusion probabilities}}\\
	7. Hard (HO) & \ding{52} & \ding{52} & \ding{56} & \ding{56} & \ding{56} & \ding{56} & \ding{56} & \ding{56} \\
	8. Soft (SO) & \ding{52} & \ding{52} & \ding{52} & \ding{52} & \ding{56} & \ding{56} & \ding{52} & \ding{52} \\
	\hline
\end{tabular} %
\end{adjustbox}
\caption{Desirable properties of prior probabilities assignments for the various priors studied, with \ding{52} indicating
that the prior has the property and \ding{56} indicating that it does not.}
\label{summary}
\end{table}

Our goal, because of the severe problems with the uniform and Jeffreys prior probabilities, was to find an objective prior distribution over model space that could be generally recommended.
Unfortunately, we were not able to arrive at a definitive conclusion.

In Table~\ref{summary} we have collected the main criteria for prior probabilities that have been studied in this paper. One could peruse this list and decide what is optimal in a given situation, but that is not realistic. So we provide our recommendations.

We definitely do not recommend the uniform prior because it violates SP and multiplicity control.
We do not recommend the Jeffreys prior because it badly violates SP, with the full model having the
largest prior probability (which is shown in the Appendix to possibly and erroneously lead to
the full model having the highest posterior probability).
We definitely do not recommend  the CMG prior because it is much too parsimonious to be labelled objective. 

For the remaining priors, the biggest
difference is SP, with the Half-p and Half-k priors satisfying SP and the rest not. This would seem
to suggest choosing between the two, with the Half-p prior perhaps winning because of simplicity. But recall that
these two priors might be viewed as too parsimonious; killing off larger models, with the full model having
only probability $O(2^{-k}/k)$.and  $O(2^{-k}/\sqrt{k})$, respectively.

If one does not wish to so severely penalize larger models, the remaining choices are the Beta(1,2), Hierarchical
and Harmonic priors, all of which have very similar properties, except that the Harmonic does not satisfy SO. As we argued
in the previous section, there are situations where not wanting SO is natural, and then the Harmonic prior would be the choice.
If one does want SO, the choice would be between the Beta(1,2) and Hierarchical priors. The Beta(1,2) prior obviously
has simplicity in its favor, but has the somewhat negative feature of more highly favoring mid-size models 
(Table~\ref{probs_model_size}).

In conclusion, if one desires SP and is okay with larger models being rapidly killed off, chose the Half-p prior.
If one does not want to rapidly kill off large models but wants the variable inclusion probabilities to decline as
$k$ grows, choose the Harmonic prior. If one does not want to rapidly kill off large models and doesn't want
the inclusion probabilities to go to 0 as $k$ grows, use the Beta(1,2) or Hierarchial priors (with the Beta(1,2)
being simpler).

\bibliographystyle{apa}
\bibliography{merged}

\begin{thebibliography}{}

\bibitem[\protect\astroncite{Berger et~al.}{2024}]{bergetal22}
Berger, J., Garc{\'\i}a-Donato, G., Moreno, E., and Pericchi, L. (2024).
\newblock {\em Handbook of Bayesian, Fiducial, and Frequentist Inference},
  chapter Objective Bayesian Testing and Model Uncertainty.
\newblock Chapman and Hall CRC.

\bibitem[\protect\astroncite{Berger and Pericchi}{1996}]{Berg:Peri:96a}
Berger, J. and Pericchi, L. (1996).
\newblock The intrinsic {B}ayes factor for model selection and prediction.
\newblock {\em Journal of the American Statistical Association}, 91:109--122.

\bibitem[\protect\astroncite{Casella et~al.}{2014}]{Cas:Mo:Gi:14}
Casella, G., Moreno, E., and Gir{\'o}n, F. (2014).
\newblock Cluster analysis, model selection and prior distributions on models.
\newblock {\em Bayesian Analysis}, 9(3):613--658.

\bibitem[\protect\astroncite{Castillo et~al.}{2015}]{Castillo_etal_15}
Castillo, I., Schmidt-Hieber, J., and van~der Vaart, A. (2015).
\newblock Bayesian linear regression with sparse priors.
\newblock {\em Annals of Statistics}, 43:1986--2018.

\bibitem[\protect\astroncite{Clyde et~al.}{1996}]{Clyd:Desi:Parmi:96}
Clyde, M., DeSimone, H., and Parmigiani, G. (1996).
\newblock Prediction via orthogonalized model mixing.
\newblock {\em Journal of the American Statistical Association}, 91:1197--1208.

\bibitem[\protect\astroncite{Furnival and Wilson}{1974}]{Furnival1974}
Furnival, G.~M. and Wilson, R.~W. (1974).
\newblock Regressions by leaps and bounds.
\newblock {\em Technometrics}, 16(4):499--511.

\bibitem[\protect\astroncite{Garc{\'\i}a-Donato and Forte}{2018}]{Gar-DonFor18}
Garc{\'\i}a-Donato, G. and Forte, A. (2018).
\newblock {B}ayesian testing, variable selection and model averaging in linear
  models using {R} with {BayesVarSel}.
\newblock {\em {The R Journal}}, 10(1):155--174.

\bibitem[\protect\astroncite{Garcia-Donato and
  Martinez-Beneito}{2013}]{Ga-DoMa-Be13}
Garcia-Donato, G. and Martinez-Beneito, M.~A. (2013).
\newblock On {S}ampling strategies in {B}ayesian variable selection problems
  with large model spaces.
\newblock {\em Journal of the American Statistical Association},
  108(501):340--352.

\bibitem[\protect\astroncite{Garc{\'\i}a-Donato and
  Steel}{2021}]{GarDonSteel21}
Garc{\'\i}a-Donato, G. and Steel, M. (2021).
\newblock Bayes factors based on g-priors for variable selection.
\newblock In Tadesse, M.~G. and Vanucci, M., editors, {\em Handbook of
  {B}ayesian {V}ariable {S}election}, chapter~15. Chapman and Hall/CRC.

\bibitem[\protect\astroncite{George and McCulloch}{1997}]{Geor:Mccu:97}
George, E.~I. and McCulloch, R. (1997).
\newblock Approaches for {B}ayesian variable selection.
\newblock {\em Statistica Sinica}, 7:339--374.

\bibitem[\protect\astroncite{George and McCulloch}{1993}]{Geor:Mccu:93}
George, E.~I. and McCulloch, R.~E. (1993).
\newblock Variable selection via {G}ibbs sampling.
\newblock {\em Journal of the American Statistical Society}, 88:881--889.

\bibitem[\protect\astroncite{Jeffreys}{1961}]{Jeff:61}
Jeffreys, H. (1961).
\newblock {\em Theory of Probability}.
\newblock Oxford University Press, London.

\bibitem[\protect\astroncite{Ley and Steel}{2009}]{Ley:Steel:09}
Ley, E. and Steel, M. F.~J. (2009).
\newblock On the effect of prior assumptions in {B}ayesian model averaging with
  applications to growth regression.
\newblock {\em Journal of Applied Econometrics (Wiley InterScience, Online)},
  24(1002):651--674.

\bibitem[\protect\astroncite{Lumley}{2024}]{leaps-package}
Lumley, T. (2024).
\newblock {\em leaps: Regression Subset Selection}.
\newblock R package version 3.2.

\bibitem[\protect\astroncite{Raftery et~al.}{1997}]{Raft:Madi:Hoet:97}
Raftery, A.~E., Madigan, D., and Hoeting, J.~A. (1997).
\newblock Bayesian model averaging for regression models.
\newblock {\em Journal of the American Statistical Association}, 92:179--191.

\bibitem[\protect\astroncite{Scott and Berger}{2005}]{Scott:Berg:05}
Scott, J. and Berger, J. (2005).
\newblock An exploration of aspects of {B}ayesian multiple testing.
\newblock {\em Journal of Statistical Planning and Inference}, 136:2144--2162.

\bibitem[\protect\astroncite{Scott and Berger}{2010}]{Scott:Berg:10}
Scott, J. and Berger, J. (2010).
\newblock Bayes and empirical-{B}ayes multiplicity adjustment in the
  variable-selection problem.
\newblock {\em Annals of Statistics}, 38:2587--2619.

\bibitem[\protect\astroncite{Smith and Kohn}{1996}]{Smit:Kohn:96}
Smith, M. and Kohn, R. (1996).
\newblock Nonparametric regression using {B}ayesian variable selection.
\newblock {\em Journal of Econometrics}, 75:317--344.

\bibitem[\protect\astroncite{Villa and Lee}{2020}]{ViLee20}
Villa, C. and Lee, J. (2020).
\newblock A loss-based prior for variable selection in linear regression
  methods.
\newblock {\em Bayesian Analysis}, 15(2):533--558.

\bibitem[\protect\astroncite{Wilson et~al.}{2010}]{Wils:Iver:Clyd:Etal:10}
Wilson, M.~A., Iversen, E.~S., Clyde, M.~A., Schmidler, S.~C., and Schildkraut,
  J.~M. (2010).
\newblock Bayesian model search and multilevel inference for {SNP} association
  studies.
\newblock {\em Annals of Applied Statistics}, 4:1342--1364.

\bibitem[\protect\astroncite{Womack et~al.}{2025}]{WoTay-RoFue25}
Womack, A., {Taylor-Rodriguez}, D., and Fuentes, C. (2025).
\newblock The matryoshka doll prior -- principled multiplicity correction in
  bayesian model comparison.
\newblock Technical report, arXiv.

\bibitem[\protect\astroncite{Zurriaga et~al.}{2011}]{Zuetal11}
Zurriaga, O., Perez-Panades, J., Izquiero, J., Gil, M., Anes, Y., Qui{\~n}ones,
  C., Margolles, M., Lopez-Maside, A., Vega-Alonso, A.~T., and Miralles, M.~T.
  (2011).
\newblock Factors associated with childhood obesity in {S}pain. the {OBICE}
  study: a case--control study based on sentinel networks.
\newblock {\em Public Health Nutrition}, 14(6):1105--113.

\end{thebibliography}

\appendix

\section{Formal variable selection and a discriminating example}

\subsection{Variable selection in the General linear model}\label{glm.formulas}
Let $X_0$ and $X$ be matrices of dimension $n\times k_0$ and $n\times k$ respectively. The columns in $X$ are the observed values of each of the entertained variables while those in $X_0$ contains covariates common to all models (e.g., an intercept). Assume $(X_0\, X)$ has rank $k_0+k$  and let $\n y$ be the $n$-dimensional vector of observations.

For $\n\gamma\ne \n 0$, $M_{\bm\gamma}$ is
$$
\n y=X_0\n\beta_0+X_{\bm\gamma}\n\beta_{\bm\gamma}+\sigma\n\epsilon,\,\,
\n\epsilon\sim N_n(\n 0,I_n),
$$
where $X_\gamma$ is the sub-matrix of $X$ obtained by selecting the columns having a one in ${\bm\gamma}$. In this paper we use the intrinsic priors derived in \cite{bergetal22}. These lead to a Bayes factor of $M_{\bm\gamma}$ to $M_0$:
\begin{equation}\label{Bgamma}
B_{{\bm\gamma}}=\int_0^\infty \Big(1+g \frac{SSE_{\bm\gamma}}{SSE_0}\Big)^{-(n-k_0)/2}(1+g)^{(n-|{\bm\gamma}|-k_0)/2}\, p^I(g)\, dg,
\end{equation}
where $p^I(g)$ is the mixing (density) function
$$
p^I(g)=\Big(\frac{n-k_0}{|{\bm\gamma}|+1}\Big)^{1/2}\frac{1}{g\pi}\Big(g-\frac{n-k_0}{|{\bm\gamma}|+1}\Big)^{-1/2},\,\,\,\, g>\frac{n-k_0}{|{\bm\gamma}|+1},
$$
and $SSE_{\bm\gamma}$ is the sum of squared errors for $M_{\bm\gamma}$. Posterior probabilities of models are $Pr(M_{\bm\gamma}\mid\n y)\propto B_{{\bm\gamma}}Pr(M_{\bm\gamma})$.

In our numerical examples computations are assisted by the {\tt BayesVarSel} package of R \cite{Gar-DonFor18} which incorporates the above Bayes factor with the option {\tt prior.betas=}``{\tt IHG}''.
When enumeration is feasible (say $k\le 25$) the package provides exact values of the posterior inclusion probabilities; the highest posterior probabilty model (HPM) and the median probability model (MPM), given, respectively, by
$$
\mbox{HPM}=\arg\max_{\bm\gamma} Pr(M_{\bm\gamma}\mid\n y), \quad \mbox{MPM}=M_{(1_{\{Pr(x_1\mid y)>0.5\}},\ldots,1_{\{Pr(x_k\mid y)>0.5\}})},
$$
where $1_A$ is the indicator function and
\begin{equation}
\label{pip}
Pr(x_j\mid\n y)=\sum_{{\bm\gamma}:\gamma_j=1}Pr(M_{\bm\gamma}\mid\n y),\,(1\le j\le k),
\end{equation}
are the posterior inclusion probabilities. When enumeration is not possible, posterior inclusion probabilities (and the MPM) are estimated using the Gibbs sampling algorithm studied in \cite{Ga-DoMa-Be13}. This method samples models from their posterior distribution and so more probable models are expected to be sampled. Nevertheless, the algorithm is not specifically designed to identify the HPM. To identify the HPM we use the branch-and-bound method proposed by \cite{Furnival1974} and incorporated in the R package {\tt leaps} \citep{leaps-package} with the function {\tt regsubsets}. This algorithm identifies, for each possible dimension $1\le d\le k$, the model with the lowest SSE, say $M_{{\bm\gamma}^\star(d)}$. Clearly, if $B_{{\bm\gamma}}$ decreases as a function of $SSE_{\bm\gamma}$ (for fixed $n, k_0$, $SSE_0$ and $|{\bm\gamma}|$), then $M_{{\bm\gamma}^\star(d)}$ is also the model with the largest posterior probability of dimension $d$ since $Pr(M_{\bm\gamma})$ is constant for all models sharing the same dimension.
From here, the HPM can be easily obtained as
$$
HPM=\arg\max_{1\le d\le  k}B_{{\bm\gamma}^\star(d)}h(d).
$$
In the obesity example with $k=47$ the HPM is identified in less than 5 minutes with a standard personal computer.
In the following lemma we show the required monotonicity of $B_{{\bm\gamma}}$.

\begin{lemma}
Using  (\ref{Bgamma}), consider $B_{{\bm\gamma}}(SSE)$ as a function of $SSE$ for fixed $n, k_0$, $SSE_0$ and $|{\bm\gamma}|$. If $SSE_1\le SSE_2$ then $B_{{\bm\gamma}}(SSE_1)\ge B_{{\bm\gamma}}(SSE_2)$.
\end{lemma}
\begin{proof}
It suffices to show that the inequality holds for the factor $(1+g \frac{SSE}{SSE_0})^{-(n-k_0)/2}$ in (\ref{Bgamma}), which follows from
$$
\frac{\partial}{\partial\, SSE}\, \Big(1+g \frac{SSE}{SSE_0}\Big)^{-(n-k_0)/2}<0 \,.
$$
\end{proof}

\subsection{ A discriminating example}
\label{sec.obesity}
We consider the $k=47$ variable obesity example for all prior model choices.
Bayesian model uncertainty was implemented with these eight prior model probability choices
and the intrinsic prior for the model parameters. Reported in Table \ref{obesity} are the
posterior inclusion probabilities, defined in (\ref{pip}), for each of the possible variables. 

Again, the most striking feature of this table is that the HPM (denoted M in the table) is the full model. This is not scientifically viable and indicates that the Jeffreys prior, giving maximum probability to the full model, is problematic.
Surprisingly this is
also the case for the Harmonic prior probabilities, suggesting that it is also
not parsimonious enough. 

The posterior inclusion probabilities do not vary much across the choice of prior	probabilities, except for the Uniform prior, where the inclusion probabilities are systematically higher, and the CMG prior with substantially lower inclusion probabilities. In contrast, the HPM (denoted H in the table) varies considerably over the choice of prior probabilities. The Uniform prior resulted in the HPM
containing 15 variables, while the other choices of prior probabilities resulted
in smaller 8 to 12 variable models (and mostly the same variables). Finally, and quite curiously, the MPM (denoted M in the table) is quite stable over
all prior probability choices, never differing by more than two variables, except for the CMG prior.
This is another piece of evidence (from a large literature not reported here) for favoring the MPM over the HPM.

\begin{figure}[t!]
\begin{center}
\includegraphics[scale=0.4]{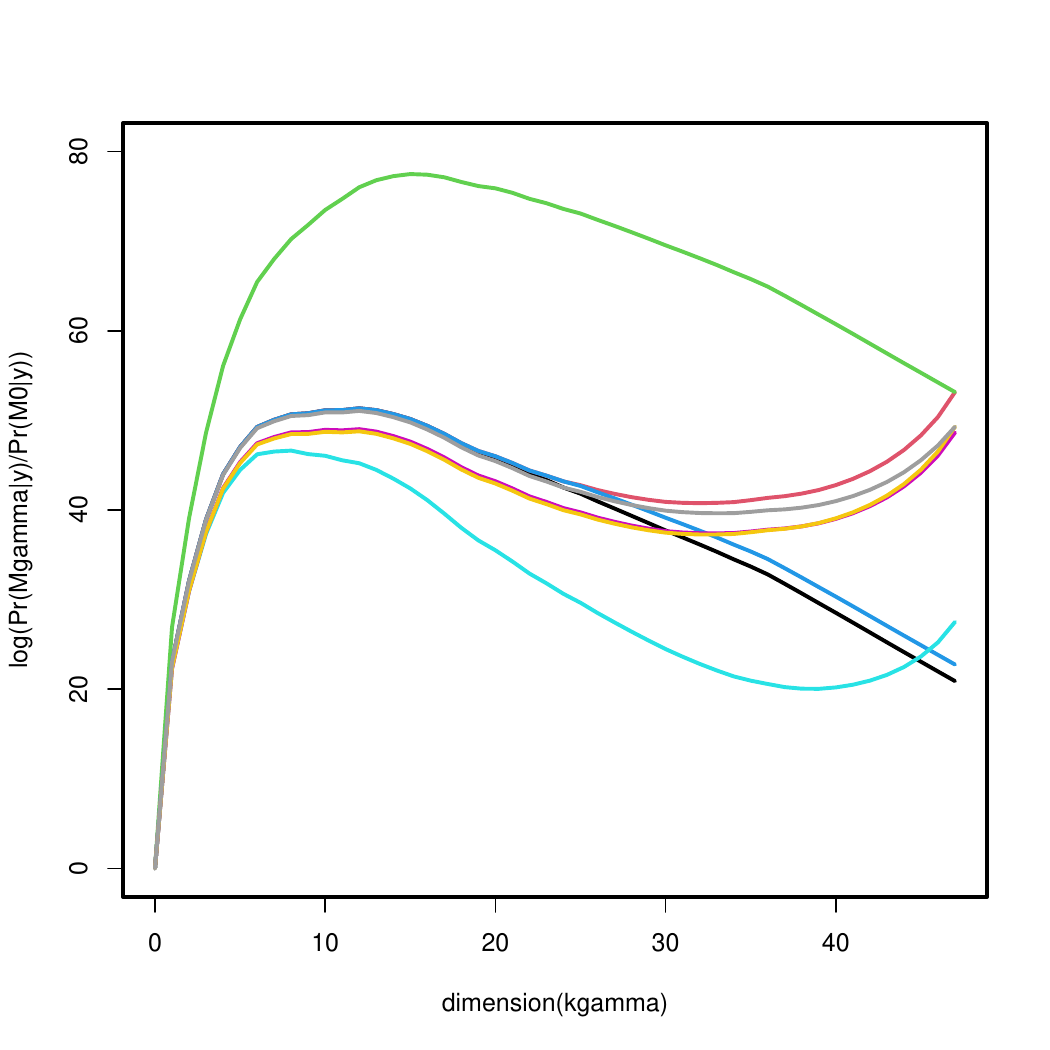}
\end{center}
\caption{For the obesity dataset, $k=47$, the log of ratio of the posterior probability of the most probable model of each dimension to the null model for different priors. Half-$p$ (black); Jeffreys (red), Uniform (Green), Half-$k$ (blue); CMG (Cyan); Hierarchical (magenta); Beta(1,2) (gray); Harmonic (yellow).}
\label{fig.allposteriors}
\end{figure}
Figure~\ref{fig.allposteriors} presents, for this example, the log of the ratio of the posterior probability of the most probable model of each dimension to the null model, for the different prior probabilities. All the priors exhibit a bimodal behavior, except for the Uniform, Half-$k$ and
Half-$p$ priors. A curiosity is that the Hierarchical and Harmonic prior probabilities seem to exactly track each other, yet the Harmonic ended up with the HPM being the full model and the Hierarchical did not.
\begin{table}[t!]
\centering
\begin{adjustbox}{max width=\textwidth}
\vspace{-3cm}
\begin{tabular}{l|cccccccc}
Variable & Jeffreys & Uniform & Half-$k$ & Half-$p$ & Hierarchical & Harmonic & Beta(1,2) & CMG\\
\hline                                                     
ActivDepor     & 0.67H/M & 0.70H/M & 0.63H/M & 0.61H/M & 0.63H/M & 0.63H/M & 0.64H/M & 0.40-/- \\
Chuches        & 1.00H/M & 1.00H/M & 1.00H/M & 1.00H/M & 1.00H/M & 1.00H/M & 1.00H/M & 1.00H/M \\ 
CincoComidas   & 0.54H/M & 0.56-/M & 0.51-/M & 0.48-/- & 0.51-/M & 0.51H/M & 0.51-/M & 0.29-/- \\
Daceite        & 0.22H/- & 0.23-/- & 0.18-/- & 0.15-/- & 0.19-/- & 0.20H/- & 0.19-/- & 0.05-/- \\
Dcereal        & 0.34H/- & 0.34-/- & 0.28-/- & 0.25-/- & 0.30-/- & 0.29H/- & 0.30-/- & 0.09-/- \\
Desayuno       & 0.98H/M & 0.98H/M & 0.98H/M & 0.97H/M & 0.98H/M & 0.98H/M & 0.98H/M & 0.97H/M \\
Dgalleta       & 0.30H/- & 0.31-/- & 0.23-/- & 0.20-/- & 0.25-/- & 0.25H/- & 0.26-/- & 0.07-/- \\
Dislipemias    & 0.56H/M & 0.59-/M & 0.51-/M & 0.45-/- & 0.51-/M & 0.51H/M & 0.52-/M & 0.22-/- \\
Dleche         & 0.28H/- & 0.29-/- & 0.23-/- & 0.21-/- & 0.25-/- & 0.25H/- & 0.25-/- & 0.10-/- \\
Dotros         & 0.42H/- & 0.46-/- & 0.37-/- & 0.33-/- & 0.38-/- & 0.38H/- & 0.39-/- & 0.12-/- \\
Dpan           & 0.31H/- & 0.31-/- & 0.25-/- & 0.22-/- & 0.26-/- & 0.26H/- & 0.26-/- & 0.08-/- \\
Dzumoenv       & 0.82H/M & 0.85H/M & 0.80H/M & 0.78H/M & 0.81H/M & 0.81H/M & 0.81H/M & 0.60H/M \\
Dzumonat       & 0.21H/- & 0.21-/- & 0.16-/- & 0.14-/- & 0.17-/- & 0.18H/- & 0.18-/- & 0.05-/- \\
Faperitivos    & 0.37H/- & 0.38-/- & 0.30-/- & 0.27-/- & 0.32-/- & 0.32H/- & 0.32-/- & 0.10-/- \\
Farroz         & 0.29H/- & 0.30-/- & 0.24-/- & 0.21-/- & 0.25-/- & 0.25H/- & 0.26-/- & 0.09-/- \\
Fcarnes        & 0.38H/- & 0.36-/- & 0.31-/- & 0.28-/- & 0.32-/- & 0.33H/- & 0.33-/- & 0.13-/- \\
Fchucherias    & 0.33H/- & 0.35-/- & 0.27-/- & 0.23-/- & 0.28-/- & 0.29H/- & 0.30-/- & 0.08-/- \\
Fdulces        & 0.99H/M & 0.99H/M & 0.99H/M & 0.99H/M & 0.99H/M & 0.99H/M & 0.99H/M & 0.98H/M \\
Ffiambres      & 0.24H/- & 0.24-/- & 0.18-/- & 0.16-/- & 0.20-/- & 0.20H/- & 0.21-/- & 0.06-/- \\
Ffritos        & 0.28H/- & 0.30-/- & 0.23-/- & 0.20-/- & 0.25-/- & 0.25H/- & 0.25-/- & 0.07-/- \\
Ffruta         & 0.21H/- & 0.22-/- & 0.17-/- & 0.14-/- & 0.18-/- & 0.18H/- & 0.19-/- & 0.05-/- \\
Fhuevos        & 0.64H/M & 0.67-/M & 0.61-/M & 0.57-/M & 0.61H/M & 0.61H/M & 0.61H/M & 0.37-/- \\
Flacteos       & 0.38H/- & 0.40-/- & 0.33-/- & 0.29-/- & 0.34-/- & 0.33H/- & 0.35-/- & 0.13-/- \\
Flegumbres     & 0.31H/- & 0.32-/- & 0.26-/- & 0.22-/- & 0.28-/- & 0.27H/- & 0.28-/- & 0.10-/- \\
Fpan           & 0.57H/M & 0.62-/M & 0.53-/M & 0.50-/M & 0.54-/M & 0.54H/M & 0.56-/M & 0.28-/- \\
Fpescado       & 0.22H/- & 0.21-/- & 0.17-/- & 0.14-/- & 0.19-/- & 0.18H/- & 0.19-/- & 0.05-/- \\
Fprecocina     & 0.23H/- & 0.25-/- & 0.18-/- & 0.16-/- & 0.20-/- & 0.20H/- & 0.21-/- & 0.06-/- \\
Frefrescos     & 0.44H/- & 0.47-/- & 0.37-/- & 0.34-/- & 0.39-/- & 0.38H/- & 0.39-/- & 0.14-/- \\
Fruta          & 0.21H/- & 0.21-/- & 0.16-/- & 0.13-/- & 0.18-/- & 0.17H/- & 0.17-/- & 0.05-/- \\
Fverduras      & 0.42H/- & 0.45-/- & 0.34-/- & 0.29-/- & 0.36-/- & 0.35H/- & 0.38-/- & 0.08-/- \\
HorasPCsem1    & 0.28H/- & 0.28-/- & 0.22-/- & 0.19-/- & 0.24-/- & 0.24H/- & 0.24-/- & 0.07-/- \\
HorasPCsem2    & 0.20H/- & 0.21-/- & 0.16-/- & 0.14-/- & 0.18-/- & 0.17H/- & 0.18-/- & 0.05-/- \\
HorasTVsem2    & 0.33H/- & 0.33-/- & 0.29-/- & 0.26-/- & 0.30-/- & 0.29H/- & 0.30-/- & 0.15-/- \\
HoraSuenyo     & 0.32H/- & 0.34-/- & 0.27-/- & 0.23-/- & 0.29-/- & 0.28H/- & 0.29-/- & 0.08-/- \\
HTA            & 0.30H/- & 0.32-/- & 0.25-/- & 0.21-/- & 0.26-/- & 0.26H/- & 0.27-/- & 0.08-/- \\
LactMater      & 0.22H/- & 0.22-/- & 0.17-/- & 0.15-/- & 0.19-/- & 0.19H/- & 0.20-/- & 0.05-/- \\
MadreObesa     & 0.99H/M & 0.99H/M & 0.99H/M & 0.99H/M & 0.99H/M & 0.99H/M & 0.99H/M & 0.99H/M \\
NumHnosOb      & 0.89H/M & 0.91H/M & 0.88H/M & 0.87H/M & 0.88H/M & 0.88H/M & 0.89H/M & 0.73H/M \\
OtrosPatol     & 0.20H/- & 0.21-/- & 0.16-/- & 0.13-/- & 0.17-/- & 0.17H/- & 0.18-/- & 0.05-/- \\
PadreObeso     & 1.00H/M & 1.00H/M & 1.00H/M & 1.00H/M & 1.00H/M & 1.00H/M & 1.00H/M & 1.00H/M \\
PesoNac        & 0.84H/M & 0.88H/M & 0.81H/M & 0.79H/M & 0.81H/M & 0.80H/M & 0.82H/M & 0.46-/- \\
ProbOsteo      & 0.85H/M & 0.89H/M & 0.83H/M & 0.80H/M & 0.83H/M & 0.82H/M & 0.83H/M & 0.59-/M \\
ProbPsico      & 0.31H/- & 0.32-/- & 0.25-/- & 0.20-/- & 0.28-/- & 0.26H/- & 0.28-/- & 0.07-/- \\
ProbResp       & 0.29H/- & 0.30-/- & 0.23-/- & 0.21-/- & 0.26-/- & 0.26H/- & 0.26-/- & 0.08-/- \\
TallaNac       & 0.92H/M & 0.95H/M & 0.91H/M & 0.89H/M & 0.90H/M & 0.90H/M & 0.91H/M & 0.61-/M \\
TVDiario       & 0.87H/M & 0.88H/M & 0.85H/M & 0.84H/M & 0.86H/M & 0.86H/M & 0.86H/M & 0.77H/M \\
Verduras       & 0.24H/- & 0.24-/- & 0.18-/- & 0.15-/- & 0.20-/- & 0.21H/- & 0.21-/- & 0.05-/- \\
\hline                                                                                         
$|HPM|/|MPM|$  & 47/16  & 15/16     & 12/16     & 12/14     & 12/16     & 47/16     & 12/16     & 8/10      \\
\end{tabular}   
\end{adjustbox}
\label{obesity}
\caption{For the obesity dataset with $k=47$ variables and $n=591$, posterior inclusion probabilities obtained with different prior distributions $P(M_{\n\gamma})$. It is also indicated whether the variable belongs to the HPM (designated by H) and/or to the MPM (designated by M). The last row is the dimension of the HPM and the MPM.}
\end{table}


\end{document}